\documentclass[twoside,11pt,final]{entics}
\usepackage{enticsmacro}

\usepackage{url}

\usepackage{bussproofs,varwidth}

\usepackage{prooftree}
\usepackage{amsmath,amssymb,stmaryrd,wrapfig} 
\usepackage{caption}
\usepackage{subcaption}
\usepackage{mleftright}

\usepackage{tikz}
\tikzstyle{bullet}=[circle,fill=black,minimum size=0.1cm,draw,inner sep=0cm]
\usetikzlibrary{trees}
\usetikzlibrary{shapes.geometric, arrows}
\tikzstyle{every picture}+=[remember picture,baseline]

\newcommand{\blue}[1]{{\color{blue}#1}}

\usepackage{cleveref}

\crefname{section}{Section}{Sections}
\crefname{subsection}{Section}{Sections}
\crefname{appendix}{Appendix}{Appendices}
\crefname{figure}{Figure}{Figures}
\crefname{table}{Table}{Tables}
\crefname{equation}{}{}
\creflabelformat{equation}{#2(#1)#3}

\crefname{theorem}{Theorem}{Theorems}
\crefname{definition}{Definition}{Definitions}

\newcommand{\code}[1]{\texttt{#1}}

\newcommand{\proofcase}[1]{\medskip\noindent{\bf Case }{#1}{\bf :}}

\newcommand{\PHL}{\ensuremath{\mathtt{PHL}}}
\newcommand{\THL}{\ensuremath{\mathtt{THL}}}
\newcommand{\TRHL}{\ensuremath{\mathtt{TRHL}}}
\newcommand{\PRHL}{\ensuremath{\mathtt{PRHL}}}

\newcommand{\var}{\mathrm{Var}}

\newcommand{\val}{\mathrm{Val}}
\newcommand{\cmd}[1]{\code{#1}}
\newcommand{\cskip}{\cmd{skip}}

\newcommand{\assign}[2]{{#1}\,:=\,{#2}}
\newcommand{\ifelse}[3]{\cmd{if}\;{#1}\;\cmd{then}\;{#2}\;\cmd{else}\;{#3}}

\newcommand{\while}[2]{\cmd{while}\;{#1}\;\cmd{do}\;{#2}}

\newcommand{\config}[2]{\langle {#1}, {#2} \rangle}
\newcommand{\diverges}{\!\uparrow}
\newcommand{\converges}{\!\downarrow}

\newcommand{\dom}{dom}
\newcommand{\fv}{fv}

\newcommand{\trace}{\mathrm{n}}

\newcommand{\exc}{\to}

\newcommand{\state}{\mathrm{\sigma}}
\newcommand{\sem}[1]{[\![#1]\!]}
\newcommand{\posts}[2]{\mathsf{post}({#1})({#2})}

\newcommand{\htriple}[3]{\mbox{$\{#1\}\,#2\,\{#3\}$}}
\newcommand{\ihtriple}[3]{\mbox{$[#1]\,#2\,[#3]$}}

\newcommand{\cskiprule}{\mbox{(Skip)}}
\newcommand{\iskiprule}{\mbox{(RSkip)}}

\newcommand{\cyclskiponerule}{\mbox{(CSkip)}}
\newcommand{\cyclskiprule}{\mbox{(CSkip2)}}
\newcommand{\crskiponerule}{\mbox{(CRSkip)}}
\newcommand{\crskiprule}{\mbox{(CRSkip2)}}
\newcommand{\cassignrule}{\mbox{($:=$)}}
\newcommand{\iassignrule}{\mbox{(R:=)}}
\newcommand{\cycassignrule}{\mbox{(C:=)}}
\newcommand{\icycassignrule}{\mbox{(CR:=)}}
\newcommand{\csequencerule}{\mbox{(Seq)}}
\newcommand{\isequencerule}{\mbox{(RSeq)}}
\newcommand{\cconseqrule}{\mbox{($\models$)}}
\newcommand{\iconseqrule}{\mbox{(R$\models$)}}
\newcommand{\cycconseqrule}{\mbox{(C$\models$)}}
\newcommand{\crconseqrule}{\mbox{(C$\models$)}}
\newcommand{\idisjtrule}{\mbox{(RDisj)}}
\newcommand{\ccondrule}{\mbox{(If)}}
\newcommand{\icondtrule}{\mbox{(RIf1)}}
\newcommand{\icondfrule}{\mbox{(RIf2)}}

\newcommand{\icyccondtrule}{\mbox{(CRIf1)}}
\newcommand{\icyccondfrule}{\mbox{(CRIf2)}}
\newcommand{\cyccondrule}{\mbox{(CIf)}}
\newcommand{\cwhilerule}{\mbox{(Inv)}}
\newcommand{\cwhiletotalrule}{\mbox{(Inv-Total)}}
\newcommand{\iwhilerule}{\mbox{(RInv-Partial)}}
\newcommand{\iwhileskiprule}{\mbox{(RInv-0)}}
\newcommand{\iwhiletotalrule}{\mbox{(RInv-Total)}}
\newcommand{\icycwhileskiprule}{\mbox{(CRInv1)}}
\newcommand{\icycwhilerule}{\mbox{(CRInv2)}}
\newcommand{\cycwhilerule}{\mbox{(CInv)}}

\newcommand{\cycsubstrule}{\mbox{(Sub)}}
\newcommand{\substrule}{\mbox{(Sub)}}
\newcommand{\isubstrule}{\mbox{(RSub)}}

\newcommand{\indhyp}{\mbox{(IH)}}

\newcommand\figref[1]{Fig. \textcolor{black}{\ref{#1}}}

\def\conf{MFPS 2025} 	
\volume{5}			
\def\volu{5}			
\def\papno{5}	
\def\mydoi{16696}
\def\lastname{Brotherston, Le, Desai, Oda} 
\def\runtitle{Cyclic proofs in Hoare logic} 
\def\copynames{J. Brotherston, Q. L. Le, G. Desai, Y. Oda}    
\begin{document}
\begin{frontmatter}
\title{Cyclic Proofs in Hoare Logic and its Reverse}
\author{James Brotherston\thanksref{a}
}
   \author{Quang Loc Le\thanksref{a}
   }
   \author{Gauri Desai\thanksref{a}
   }		
   \author{Yukihiro Oda\thanksref{b}
   }
   \address[a]{Department of Computer Science, University College London,	
    London, United Kingdom}  							
   \address[b]{Tohoku University, Sendai, Japan}

\maketitle

\begin{abstract}
We examine the relationships between axiomatic and cyclic proof systems for the partial and total versions of Hoare logic and those of its dual, known as reverse Hoare logic (or sometimes incorrectness logic).

In the axiomatic proof systems for these logics, the proof rules for looping constructs involve an explicit loop invariant, which in the case of the total versions additionally require a well-founded termination measure.  In the cyclic systems, these are replaced by rules that simply unroll the loops, together with a principle allowing the formation of cycles in the proof, subject to a global soundness condition that ensures the well-foundedness of the circular reasoning.  Interestingly, the cyclic soundness conditions for partial Hoare logic and its reverse are similar and essentially coinductive in character, while those for the total versions are also similar and essentially inductive.  We show that these cyclic systems are sound, by direct argument, and relatively complete, by translation from axiomatic to cyclic proofs.
\end{abstract}


\end{frontmatter}

\section{Introduction}
\label{sec:intro}

Hoare logic~\cite{Hoare:69} is long established as a formal framework in which to specify and reason about correctness properties of computer programs.  Typically, these properties, written as triples $\htriple{P}{C}{Q}$, are given a \emph{partial correctness} interpretation, meaning that certain unwanted behaviours of the program $C$ cannot happen: namely those in which $C$ runs on a state satisfying the \emph{precondition} $P$ but results in a state not satisfying the \emph{postcondition} $Q$.  Sometimes one works instead with a \emph{total correctness} interpretation of judgements, in which the termination of $C$ from states satisfying $P$ is additionally guaranteed.

More recently gaining in popularity is the dual notion of \emph{reverse Hoare logic} proposed in~\cite{Edsko:SEFM:2011} for specifying and proving the presence of specified behaviours in programs, as opposed to their absence.  In the original version of reverse Hoare logic one gives a reachability or ``incorrectness'' interpretation to triples $\ihtriple{P}{C}{Q}$, meaning that certain behaviours of $C$ do happen: namely, there are computations of $C$ reaching all states satisfying $Q$ from states satisfying $P$.  O'Hearn went on to propose an extension of this logic dubbed \emph{incorrectness logic} which supports the specification of various types of program errors in the postcondition $Q$ and so can be used as the basis of automatic bug detection via proof search~\cite{OHearn:19,Le:OOPSLA:2022}. A weaker \emph{partial reachability} interpretation of reverse Hoare triples can also be given, in which validity also admits the possibility of program divergence from states satisfying $P$.

Turning to proof theory, Hoare logic was originally formulated as an axiomatic proof system with rules for each program construct and a small number of generic rules; proofs themselves are then as usual finite derivation trees of proof judgements.  Of particular note is the partial Hoare logic rule for {\tt while} loops, which requires the prover to invent a suitable \emph{loop invariant} that is preserved over any execution of the loop body (and entails the desired postcondition). Total Hoare logic possesses a similar rule, except that we must also provide a well-founded \emph{termination measure} $t$ that decreases on each iteration of the loop:
\[\begin{array}{c@{\hspace{1.5cm}}c}
\begin{prooftree}
\htriple{P \wedge B}{C}{P}
\justifies
\htriple{P}{\while{B}{C}}{P \wedge \neg B}
\using \cwhilerule
\end{prooftree}
&
\begin{prooftree}
\htriple{P \wedge B \wedge t = n}{C}{P \land t<n}
\justifies
\htriple{P}{\while{B}{C}}{P \wedge \neg B}
\using \cwhiletotalrule
\end{prooftree}
\end{array}\]
Reverse Hoare logic possesses analogous ``backwards'' versions of these proof rules in its partial and total variants, respectively (cf. Section~\ref{sec.two.hoare}). Devising suitable loop invariants and termination measures is challenging and a serious obstacle to automated proof search, and has been the subject of much work in the literature on program verification~\cite{McMillan:08}.

An alternative way of approaching these challenges is to instead adopt systems based on \emph{cyclic proof}, which were first introduced to reason about process calculi~\cite{Stirling-walker:91,Niwinski-Walukiewicz:97} and later used to develop proof systems for various logics with (co)inductively defined constructs (see e.g.~\cite{Dam-Gurov:02,Brotherston:05,Brotherston:07,Brotherston-Gorogiannis:14,Afshari-Leigh:17,Das-Pous:17,Docherty-Rowe:19,Tellez-Brotherston:19}).  Cyclic proofs are essentially non-wellfounded, regular derivation trees known as ``pre-proofs'' --- typically presented as finite trees with ``backlinks'' from open leaves to interior nodes --- subject to a \emph{global soundness condition} typically ensuring that each infinite path in the pre-proof implicitly embodies a valid argument by infinite descent~\cite{Brotherston:PhD,Brotherston-Simpson:11}.
It is perhaps less widely known than it might be that both partial and total Hoare logic can be formulated as cyclic proof systems sharing exactly the same proof rules but employing two different global soundness conditions (an observation going back at least to~\cite{Brotherston-Gorogiannis:14}).

In this paper, we present axiomatic and cyclic proof systems for the total and partial versions of Hoare logic and reverse Hoare logic, examine their commonalities and symmetries, and we show all of our systems sound and relatively complete\footnote{Note to referees: In fact we have not yet fully established the relative completeness of our axiomatic system for partial reverse Hoare logic with respect to validity, but it seems highly plausible.}. For partial correctness, we stipulate that any infinite path in the pre-proof must contain infinitely many symbolic command executions, which (with local soundness of the proof rules) ensures that any putative counterexample to $\htriple{P}{C}{Q}$, i.e. a finite execution of $C$ from $P$ not resulting in $Q$, in fact cannot be finite and therefore is not a counterexample after all.  For total correctness, we instead require that any infinite path must contain an infinite descending chain of well-founded values in the preconditions along the path.  Thus, any putative counterexample to $\htriple{P}{C}{Q}$, i.e. an execution of $C$ from $P$ that either is infinite, or finite but does not result in $Q$, in fact \emph{must} be infinite and yield an infinite descending chain of well-founded values, again a contradiction. For reverse Hoare logic, we obtain similar global soundness conditions: the condition for partial reverse Hoare logic is similar to the condition for standard partial Hoare logic, and similarly for the total versions. Soundness of the cyclic systems follows by direct arguments of the usual type for such systems (cf.~\cite{Brotherston:PhD}) and their relative completeness follows by explicit transformations from standard proofs to cyclic proofs.  The developments in our paper are shown diagrammatically in Figure~\ref{fig:map}.

The remainder of this paper is structured as follows. Section~\ref{sec:pl} describes our programming language and its semantics. Section~\ref{sec.two.hoare} presents the partial and total versions of Hoare logic and their correspondents in reverse Hoare logic, both semantically and as standard axiomatic proof systems. Section~\ref{sec:cyclic} and Section~\ref{sec:cyclic:il:hoare} present our cyclic proof systems for these logics, along with their corresponding soundness and relative completeness arguments.  Finally, Section~\ref{sec:conclusion} discusses future work and concludes.

\tikzstyle{JBnode} = [rectangle, rounded corners,
minimum width=3.2cm,
minimum height=1.4cm,
text centered,
draw=black,
fill=gray!15]

\tikzstyle{JBarrow} = [thick,rounded corners=0.5cm,->,>={stealth}]
\tikzstyle{JBdoublearrow} = [thick,rounded corners=0.5cm,<->,>=stealth]
\tikzstyle{JBdoubledashed} = [thick,dashed,<->,>=stealth]

\begin{figure}[hpt]
\centering
\begin{tikzpicture}

\node (CycTRHL) [JBnode,align=center] at (4.5,2) {cyclic provability \\ in $\TRHL$};
\node (AxTRHL)  [JBnode,align=center] at (4.5,5.5) {provability \\ in $\TRHL$};
\node (ValTRHL) [JBnode,align=center] at (4.5,9) {validity \\ in $\TRHL$};

\node (CycPRHL) [JBnode,align=center] at (11.5,2) {cyclic provability \\ in $\PRHL$};
\node (AxPRHL)  [JBnode,align=center] at (11.5,5.5) {provability \\ in $\PRHL$};
\node (ValPRHL) [JBnode,align=center] at (11.5,9) {validity \\ in $\PRHL$};

\node (CycPHL) [JBnode,align=center] at (4.5,20) {cyclic provability \\ in $\PHL$};
\node (AxPHL)  [JBnode,align=center] at (4.5,16.5) {provability \\ in $\PHL$};
\node (ValPHL) [JBnode,align=center] at (4.5,13) {validity \\ in $\PHL$};

\node (CycTHL) [JBnode,align=center] at (11.5,20) {cyclic provability \\ in $\THL$};
\node (AxTHL)  [JBnode,align=center] at (11.5,16.5) {provability \\ in $\THL$};
\node (ValTHL) [JBnode,align=center] at (11.5,13) {validity \\ in $\THL$};

\draw [JBarrow] (AxPHL) -- (CycPHL) node[midway,left] {Thm.~\ref{thm:HL_translation}};
\draw [JBarrow] (AxTHL) -- (CycTHL) node[midway,right] {Thm.~\ref{thm:HL_translation}};
\draw [JBarrow] (AxPRHL) -- (CycPRHL) node[midway,right] {Thm.~\ref{thm:RHL_translation}};
\draw [JBarrow] (AxTRHL) -- (CycTRHL) node[midway,left] {Thm.~\ref{thm:RHL_translation}};
\draw [JBdoublearrow] (ValPHL) -- (AxPHL) node[midway,left,align=left] {Prop.~\ref{prop:hl_sound} / \\ Thm.~\ref{thm:hl_complete}};
\draw [JBdoublearrow] (ValTHL) -- (AxTHL) node[midway,right,align=left] {Prop.~\ref{prop:hl_sound} / \\ Thm.~\ref{thm:hl_complete}};
\draw [JBdoublearrow] (ValTRHL) -- (AxTRHL) node[midway,left,align=left] {Prop.~\ref{prop:rhl_sound} / \\ Thm.~\ref{thm:rhl_complete}};
\draw [JBdoublearrow] (ValPRHL) -- (AxPRHL) node[midway,right,align=left] {Prop.~\ref{prop:rhl_sound} / \\ Thm.~\ref{thm:rhl_complete}};
\draw [JBarrow] (CycTHL) -- (CycPHL) node[midway,above] {Defn.~\ref{defn:cyclic_proof_HL}};
\draw [JBarrow] (AxTHL) -- (AxPHL) node[midway,above] {Remark~\ref{rem:total_partial_provable}};
\draw [JBarrow] (ValTHL) -- (ValPHL) node[midway,above] {Defn.~\ref{defn:valid_triples}};
\draw [JBarrow] (CycTRHL) -- (CycPRHL)  node[midway,above] {Defn.~\ref{defn:cyclic_proof_RHL}};
\draw [JBarrow] (AxTRHL) -- (AxPRHL) node[midway,above] {Remark~\ref{rem:reverse_total_partial_provable}};
\draw [JBarrow] (ValTRHL) -- (ValPRHL) node[midway,above] {Defn.~\ref{defn:valid_triples}};

\draw[JBarrow] (CycPHL.west) -- ++(-1.5,0) |- node[near start,left] {Thm.~\ref{thm:cyclic_HL_sound}} (ValPHL.west) ;
\draw[JBarrow] (CycTHL.east)  -- ++(1.5,0) |- node[near start, right] {Thm.~\ref{thm:cyclic_HL_sound}} (ValTHL.east)  ;
\draw[JBarrow] (CycTRHL.west) -- ++(-1.5,0) |- node[near start,left] {Thm.~\ref{thm:cyclic_RHL_soundness}} (ValTRHL.west);
\draw[JBarrow] (CycPRHL.east) -- ++(1.5,0) |- node[near start,right] {Thm.~\ref{thm:cyclic_RHL_soundness}} (ValPRHL.east);

\draw[JBdoubledashed] (ValPHL) -- (ValTRHL) node[midway,left,align=center] {dual \\ (Defn.~\ref{defn:valid_triples})};
\draw[JBdoubledashed] (ValTHL) -- (ValPRHL) node[midway,right,align=center] {dual \\ (Defn.~\ref{defn:valid_triples})};
\end{tikzpicture}
\caption{\label{fig:map} Summary of our developments. $\PHL$ (resp. $\THL$) is partial (resp. total) Hoare logic, and $\PRHL$ (resp. $\TRHL$) is partial (resp. total) reverse Hoare logic. }

\end{figure}

\section{Programming language}
\label{sec:pl}

In this section, we give the syntax and semantics of a simple {\em while} programming language (without procedures or function calls).

We assume an infinite set $\var$ of \emph{variables} ranged over by $x,y,z,\ldots$.  \emph{Expressions} are first-order terms built from variables, numerals and possibly function symbols ($+$, $\times$, etc). The syntax of \emph{Boolean conditions} $B$ and \emph{programs} $C$ is then given
by the following grammar:
\[\begin{array}{rl}
    B := & E = E \mid E \leq E \mid \neg B \mid B \wedge B   \\[1ex]
    C := & \cskip \mid \assign{x}{E} \mid C ; C \mid \ifelse{B}{C}{C} \mid
     \while{B}{C}
\end{array}\]
We write $E[E'/x]$ and $B[E'/x]$ to denote the substitution of expression $E'$ for variable $x$ in expression $E$ and Boolean condition $B$, respectively, and write $\fv(C)$ to denote the set of program variables in $C$.  We also employ standard shorthand notations such as $E \neq E$ for $\neg(E_1 = E_2)$ and $E_1 < E_2$ for $E_1 \leq E_2 \wedge \neg E_1=E_2$.

Semantically, expressions denote (program) \emph{values} in a set $\val$, which for technical convenience we take in this paper to be the natural numbers. A \emph{(program) state} is a function $\sigma: \var \rightarrow \val$.  The interpretations $\sem{E}\sigma \in \val$ of any expression $E$ and $\sem{B}\sigma \in \{\top,\bot\}$ of any Boolean condition $B$ in state $\sigma$ are then standard: $\sem{x}\sigma = \sigma(x)$, $\sem{E_1+E_2}\sigma = \sem{E_1}\sigma+\sem{E_2}\sigma$, $\sem{E_1 \leq E_2}\sigma \Leftrightarrow \sem{E_1}\sigma \leq \sem{E_2}\sigma$, and so on.  Finally, we write $\sigma[x \mapsto E]$ for the state defined exactly as $\sigma$ except that $\sigma[x \mapsto E](x) = \sem{E}\sigma$. We write $\dom(\sigma)$ to denote
the domain of $\sigma$.

We define a standard small-step operational semantics of our programs. A \emph{program configuration} is a pair $\config{C}{\sigma}$, where $C$ is a program and $\sigma$ a state.
The small-step semantics is then given by a binary relation $\exc$ on configurations, as shown in Figure~\ref{fig:prog_sem}.
 An \emph{execution} (of $C$) is a possibly infinite sequence of configurations $(\gamma_i)_{i \geq 0}$ with $\gamma_0 = \config{C}{\,\_\,}$ such that $\gamma_i \exc \gamma_{i+1}$ for all $i \geq 0$.
 We write $\exc^n$ for the $n$-step variant of $\exc$ to represent $(\gamma_i)_{0 \leq i \leq n}$ as $\gamma_0 \exc^n \gamma_n$, and $\exc^*$ for the reflexive-transitive closure of $\exc$
 to represent finite executions of arbitrary length. We write $\config{C}{\sigma}\diverges$ if $\config{C}{\sigma}$ \emph{diverges}, i.e. there is an infinite execution beginning $\config{C}{\sigma} \exc \ldots$, and $\config{C}{\sigma}\converges$ if $\config{C}{\sigma}$ \emph{converges}, meaning that there is no such execution. For any program $C$ we may harmlessly identify $C$ with $C;\cskip$ (but only finitely often).

\begin{figure}[tb]\vspace{-5mm}
\[\begin{array}{c}
\begin{prooftree}
\phantom{D}
\justifies
\config{\cskip; C}{\sigma} \exc \config{C}{\sigma}
\end{prooftree}
\qquad
\begin{prooftree}
\phantom{C}
\justifies
\config{\assign{x}{E}}{\sigma} \exc \config{\cskip}{\sigma[x \mapsto \sem{E}\sigma]}
\end{prooftree}
\qquad
\begin{prooftree}
\config{C_1}{\sigma  } \exc \config{C_1'}{\sigma'}
\justifies
\config{C_1; C_2}{\sigma} \exc \config{C_1'; C_2}{\sigma'}
\end{prooftree}
\\[1ex]\\
\begin{prooftree}
\sigma \models B
\justifies
\config{\ifelse{B}{C_1}{C_2}}{\sigma} \exc \config{C_1}{\sigma}
\end{prooftree}
\qquad
\begin{prooftree}
\sigma \models \neg B
\justifies
\config{\ifelse{B}{C_1}{C_2}}{\sigma} \exc \config{C_2}{\sigma}
\end{prooftree}
\\[1ex]\\
\begin{prooftree}
\sigma \models B
\justifies
\config{\while{B}{C}}{\sigma} \exc  \config{C; \while{B}{C}}{\sigma}
\end{prooftree}
\qquad
\begin{prooftree}
\sigma \models \neg B
\justifies
\config{\while{B}{C}}{\sigma} \exc \config{\cskip}{\sigma}
\end{prooftree}
\end{array}\]

\caption{\label{fig:prog_sem} Operational semantics of while programs.}
\end{figure}

\section{Hoare logic and reverse Hoare logic}
\label{sec.two.hoare}
In this section, we define Hoare logic and reverse Hoare logic, in both their partial and total forms, first in terms of validity and then as axiomatic proof systems.

\emph{Assertions}, ranged over by $P,Q,R,\ldots$ are standard formulas of first-order logic, whose terms at least include our expressions $E$ and whose atomic formulas at least include our Boolean conditions $B$ as given in the previous section.  We write $\sigma \models P$ to denote satisfaction of assertion $P$ by state $\sigma$, defined as usual in first-order logic, and write $\sem{P}$ as a shorthand for $\{\sigma \mid \sigma \models P\}$. Then a \emph{Hoare triple} is written as $\htriple{P}{C}{Q}$ and a \emph{reverse Hoare triple} as $\ihtriple{P}{C}{Q}$, where $C$ is a program and $P$ and $Q$ are assertions.

\begin{definition}[Validity]
\label{defn:valid_triples}
First, for a program $C$ and assertion $P$, define $\posts{C}{P}$ , the \emph{post-states} of $C$ under $P$, by
\[
   \posts{C}{P} = \{\sigma' \mid \exists\sigma.\ \sigma \models P \mbox{ and }
   \config{C}{\sigma} \exc^* \config{\cskip}{\sigma'} \}\  .
\]
We write $\posts{C}{P}\converges$ to mean that $\config{C}{\sigma}\converges$ for all $\sigma \in \sem{P}$.

Then we have notions of \emph{validity} for Hoare triples $\htriple{P}{C}{Q}$ in \emph{partial} ($\PHL$) and \emph{total} ($\THL$) Hoare logic, and for reverse Hoare triples $\ihtriple{P}{C}{Q}$ in total ($\TRHL$) and partial ($\PRHL$) reverse Hoare logic:
\begin{itemize}
    \item $\htriple{P}{C}{Q}$ is \emph{valid in $\PHL$} iff $\posts{C}{P} \subseteq \sem{Q}$;
    \item $\htriple{P}{C}{Q}$ is \emph{valid in $\THL$} iff $\posts{C}{P} \subseteq \sem{Q}$ and $\posts{C}{P}\converges$;
    \item $\ihtriple{P}{C}{Q}$ is \emph{valid in $\TRHL$} iff $\posts{C}{P} \supseteq \sem{Q}$ ;
    \item $\ihtriple{P}{C}{Q}$ is \emph{valid in $\PRHL$} iff $\posts{C}{P}\converges$ implies $\posts{C}{P} \supseteq \sem{Q}$.
\end{itemize}
\end{definition}

We note that, by definition, validity in $\THL$ immediately implies validity in $\PHL$, and validity in $\TRHL$ implies validity in $\PRHL$. By unpacking the notations in Defn.~\ref{defn:valid_triples}, validity for each of our logics can also be restated in a perhaps more convenient operational form, as follows:
\begin{description}
\item[$\PHL$:] $\htriple{P}{C}{Q}$ is valid iff $\forall \sigma,\sigma'.$ if $\sigma \models P$ and $\config{C}{\sigma} \exc^* \config{\cskip}{\sigma'}$ then $\sigma' \models Q$;
\item[$\THL$:] $\htriple{P}{C}{Q}$ is valid iff $\forall \sigma,\sigma'$ if $\sigma \models P$ then $\config{C}{\sigma}\converges$, and if $\config{C}{\sigma} \exc^* \config{\cskip}{\sigma'}$ then $\sigma' \models Q$;
\item[$\TRHL$:] $\ihtriple{P}{C}{Q}$ is valid iff $\forall \sigma'.$ if $\sigma' \models Q$ then $\exists \sigma.\ \sigma \models P$ and $\config{C}{\sigma} \exc^* \config{\cskip}{\sigma'}$;
\item[$\PRHL$:] $\ihtriple{P}{C}{Q}$ is valid iff $\forall \sigma'.$ if $\sigma' \models Q$ then $\exists \sigma.\ \sigma \models P$ and either $\config{C}{\sigma} \exc^* \config{\cskip}{\sigma'}$ or $\config{C}{\sigma}\diverges$.
\end{description}

\subsection{Provability in partial and total Hoare logic}
\label{sec:hoare_axiom}

Provability in Hoare logic is defined as derivability using the rules in Fig.~\ref{fig:hoare_rules}. The partial version $\PHL$ of Hoare logic uses the rule $\cwhilerule$, which requires the prover to find a \emph{loop invariant} that is preserved by the loop body at each iteration, while the total version $\THL$ omits this rule in favour of the rule $\cwhiletotalrule$ which additionally requires the prover to find a suitable well-founded termination measure $t$ that decreases at each iteration\footnote{Note that the termination measure may not become negative because all our values are natural numbers; but, if desired, one can add for safety the requirement that $P \models t \geq 0$.}. The remaining rules are standard~\cite{Hoare:69,Apt-Olderog:19}, except that our assignment axiom is written in the forward style from Floyd~\cite{Floyd:67}.

\begin{figure}[tbh]
\[\begin{array}{c}
\begin{prooftree}
\phantom{P}
\justifies
\htriple{P}{\cskip}{P}
\using \cskiprule
\end{prooftree}
\quad\quad
\begin{prooftree}
\phantom{P}
\justifies
\htriple{P}{x := E}{P[x'/x] \land x = E[x'/x]}
\using \cassignrule
\end{prooftree}
\quad\quad
\begin{prooftree}
\htriple{P}{C_1}{R} \phantom{w} \htriple{R}{C_2}{Q}
\justifies
\htriple{P}{C_1;C_2}{Q}
\using \csequencerule
\end{prooftree}
\\ \\
\begin{prooftree}
P' \models P \quad \htriple{P}{C}{Q} \quad Q \models Q'
\justifies
\htriple{P'}{C}{Q'}
\using \cconseqrule
\end{prooftree}
\qquad
\begin{prooftree}
\htriple{P \wedge B}{C_1}{Q} \phantom{w}
\htriple{P \wedge \neg B}{C_2}{Q}
\justifies
\htriple{P}{\ifelse{B}{C_1}{C_2}}{Q}
\using \ccondrule
\end{prooftree}
\\[1ex] \\
\begin{prooftree}
\htriple{P \wedge B}{C}{P}
\justifies
\htriple{P}{\while{B}{C}}{P \wedge \neg B}
\using \cwhilerule
\end{prooftree}
\qquad
\begin{prooftree}
\htriple{P \wedge B \wedge t = n}{C}{P \land t<n}
\justifies
\htriple{P}{\while{B}{C}}{P \wedge \neg B}
\using \cwhiletotalrule
\end{prooftree}
\end{array}\]
\caption{\label{fig:hoare_rules} Proof rules for Hoare triples.  Partial Hoare logic $PHL$ uses the rule $\cwhilerule$ and total Hoare logic $THL$ uses $\cwhiletotalrule$, with all other rules shared. Variable $x'$ is fresh in rule $\cassignrule$ and variable $t$ is fresh in $\cwhiletotalrule$.}
\end{figure}

\begin{remark}
\label{rem:total_partial_provable}
Given the premise of $\cwhiletotalrule$, namely $\htriple{P \wedge B \wedge t = n}{C}{P \land t<n}$, it is easy to derive the premise of $\cwhilerule$, namely $\htriple{P \wedge B}{C}{P}$: on the one hand we have $P \land t < n \models P$, and on the other we have $P \wedge B \models P \wedge B \wedge t=n$, because $n$ is fresh.  Thus, provability in total Hoare logic $\THL$ immediately implies provability in its partial version $\PHL$ as well.
\end{remark}

The following results are well known in the literature.

\begin{proposition}[Soundness]
\label{prop:hl_sound}
If \htriple{P}{C}{Q} is provable in $\PHL$ (resp. $\THL$) then it is valid in $\PHL$ (resp. $\THL$).
\end{proposition}

\begin{theorem}[Relative completeness~\cite{Cook:1978}]
\label{thm:hl_complete}
Assuming an oracle for logical entailment between assertions, if \htriple{P}{C}{Q} is valid in $\PHL$ (resp. $\THL$) then it is provable in $\PHL$ (resp. $\THL$).
\end{theorem}

\begin{example}
We use the program $\while{x>0} {x := x - 2 ;}$ as a running example. This program has invariant $x\geq 0 \land x \% 2 = 0$ where the modulo operator $\%$ is used to specify even numbers. Here, we show a $\THL$ proof of the total correctness triple
$\htriple{x\geq 0 \land x \% 2 = 0}{\while{x>0} {x := x - 2 ;}}{x = 0}$.

\begin{footnotesize}
\[\begin{prooftree}
\[
\[\[
\justifies
\htriple{x\geq 0 \land x \% 2 = 0 \land x<2n}{\cskip}{x\geq 0 \land x \% 2 = 0 \land x<2n} \tikz \node (bud2) {}; \using  \cskiprule
\]
\justifies
\htriple{x'\geq 0 \land x' \% 2 = 0 \land x'>0  \land x'=2n \land x=x'-2}{\cskip}{x\geq 0 \land x \% 2 = 0 \land x<2n} \using \cconseqrule
\]
\justifies
\htriple{x\geq 0 \land x \% 2 = 0 \land x>0 \land x=2n}{x := x - 2;}{x\geq 0 \land x \% 2 = 0 \land x<2n}  \using \cassignrule
\]
\justifies
\htriple{x \geq 0 \land x \% 2 = 0}{\while{x>0} {x := x - 2 ;}}{x = 0} \tikz \node (comp2) {};
\using \cwhiletotalrule
\end{prooftree}
\]
\end{footnotesize}
In this proof, we use the termination measure $t = x/2$.
\end{example}

\subsection{Provability in partial and total reverse Hoare logic}
\label{sec:il:hoare}

We give our proof rules for reverse Hoare triples $\ihtriple{P}{C}{Q}$ in Figure~\ref{fig:ihoare_rules}. Here, $\PRHL$ uses rules $\iwhileskiprule$ and $\iwhilerule$, while $\TRHL$ omits these in favour of $\iwhiletotalrule$. Note that the conditional rules $\icondtrule$ and $\icondfrule$, like those in \cite{Le:OOPSLA:2022}, have stronger preconditions than the ones in \cite{Edsko:SEFM:2011,OHearn:19}.

\begin{figure}[tbh]
\[\begin{array}{c}
\begin{prooftree}
\justifies
\ihtriple{P}{\cskip}{P}
\using \iskiprule
\end{prooftree}
\qquad
\begin{prooftree}
\justifies
\ihtriple{P}{x := E}{P[x'/x] \land x = E[x'/x]}
\using \iassignrule
\end{prooftree}
\\ [1ex]\\
\begin{prooftree}
\ihtriple{P}{C_1}{R} \phantom{w}
\ihtriple{R}{C_2}{Q}
\justifies
\ihtriple{P}{C_1;C_2}{Q}
\using \isequencerule
\end{prooftree}
\qquad
\begin{prooftree}
P \models P' \quad \ihtriple{P}{C}{Q} \quad
Q' \models Q
\justifies
\ihtriple{P'}{C}{Q'}
\using \iconseqrule
\end{prooftree}
\\[1ex] \\
\begin{prooftree}
\ihtriple{P}{C}{Q}
\justifies
\ihtriple{P[t/z]}{C}{Q[t/z]}
\using  \isubstrule
\end{prooftree}
\qquad
\begin{prooftree}
  \ihtriple{P_1}{C}{Q_1} \quad
\ihtriple{P_2}{C}{Q_2}
\justifies
\ihtriple{P_1 \lor P_2}{C}{Q_1 \lor Q_2}
\using \idisjtrule
\end{prooftree}
\\[1ex] \\
\begin{prooftree}
\ihtriple{P \wedge B}{C_1}{Q}
\justifies
\ihtriple{P \wedge B}{\ifelse{B}{C_1}{C_2}}{Q}
\using \icondtrule
\end{prooftree}
 \qquad
\begin{prooftree}
\ihtriple{P \wedge \neg B}{C_2}{Q}
\justifies
\ihtriple{P \wedge \neg B}{\ifelse{B}{C_1}{C_2}}{Q}
\using \icondfrule
\end{prooftree}
\\[1ex] \\
\begin{prooftree}
\phantom{T}
\justifies
\ihtriple{P \wedge \neg B}{\while{B}{C}}{P \wedge \neg B}
\using \iwhileskiprule
\end{prooftree}
\qquad
\begin{prooftree}
\ihtriple{P \wedge B}{C}{P}
\justifies
\ihtriple{P}{\while{B}{C}}{P \wedge \neg B}
\using \iwhilerule
\end{prooftree}
\\[1ex]\\
\begin{prooftree}
\ihtriple{P \wedge B \wedge t<n}{C}{P \wedge t=n}
\justifies
\ihtriple{P}{\while{B}{C}}{P \wedge \neg B}
\using (n>0) \iwhiletotalrule
\end{prooftree}
\end{array}\]
\caption{\label{fig:ihoare_rules} Proof rules for reverse Hoare triples.  Partial reverse Hoare logic $PRHL$ uses the rules $\iwhilerule$, while total reverse Hoare logic $TRHL$ uses $\iwhiletotalrule$, with all other rules shared. Variable $x'$ is fresh in rule $\iassignrule$ and variable $n$ is fresh in $\iwhiletotalrule$. }
\end{figure}

\begin{remark}
\label{rem:reverse_total_partial_provable}
Similarly to the observation in Remark~\ref{rem:total_partial_provable}, provability in $\TRHL$ immediately implies provability in $\PRHL$ as well.
\end{remark}

The proofs of soundness and relative completeness in the total version of reverse Hoare logic appear in \cite{Edsko:SEFM:2011} and also can be seen as the special \blue{\em ok} case of incorrectness logic in~\cite{OHearn:19}.

\begin{proposition}[Soundness]
\label{prop:rhl_sound}
If $\ihtriple{P}{C}{Q}$ is provable in $\PRHL$ (resp. $\TRHL$) then it is valid in $\PRHL$ (resp. $\TRHL$).
\end{proposition}

\begin{proof}
By rule induction on the inference rules in \figref{fig:ihoare_rules}.  Generally speaking, for the shared rules, their soundness of $\TRHL$ is a special case of the more general soundness for $\PRHL$. Here we just show the more interesting cases.

\proofcase{\isequencerule}
We first consider the case of $\PRHL$. Then, supposing $\state'\models Q$, we require to find $\sigma$ with $\sigma \models P$ such that either $\config{C_1;C_2}{\sigma} \exc^* \config{\cskip}{\sigma'}$, or $\config{C_1;C_2}{\sigma}\diverges$. By validity of the right premise, we obtain $\sigma''$ with $\sigma''\models R$ and for which there are two possibilities:
\begin{itemize}
\item We have $\config{C_2}{\sigma''} \exc^* \config{\cskip}{\sigma'}$.  By validity of the left premise there are two further subcases.  In the first we have $\sigma$ with $\sigma\models P$ and $\config{C_1}{\sigma}\exc^*\config{\cskip}{\sigma''}$.  Then by the operational semantics we obtain $\config{C_1;C_2}{\sigma} \exc^* \config{\cskip;C_2}{\sigma''} \exc^* \config{\cskip}{\sigma'}$ and are done.  In the second we have $\sigma$ with $\sigma \models P$ and $\config{C_1}{\sigma}\diverges$, in which case $\config{C_1;C_2}{\sigma}\diverges$ as well.

\item We have $\config{C_2}{\sigma''}\diverges$.  Then by the left premise we have $\sigma$ with $\sigma \models P$ and either $\config{C_1}{\sigma}\exc^*\config{\cskip}{\sigma''}$ or $\config{C_1}{\sigma}\diverges$.  Either way, we have $\config{C_1;C_2}{\sigma}\diverges$ and so are done.
\end{itemize}
For $\TRHL$, the argument above is simplified: we simply ignore all the subcases involving divergence.

\proofcase{\icondtrule,\icondfrule}
These rules are symmetric; we just consider the first.
In $\PRHL$, suppose $\state' \models Q$. By validity of the rule premise, we obtain $\sigma$ with $\sigma \models P \wedge B$ obeying one of two possibilities.  First, if $\config{C_1}{\sigma}\exc^*\config{\cskip}{\sigma'}$ then $\config{\ifelse{B}{C_1}{C_2}}{\sigma}\exc^*\config{\cskip}{\sigma'}$ as well, because $\sigma \models B$. Otherwise, if $\config{C_1}{\sigma}\diverges$ then $\config{\ifelse{B}{C_1}{C_2}}{\sigma}\diverges$ too. (For $\TRHL$, we simply ignore the second case.)

\proofcase{\iwhileskiprule} Supposing $\state' \models P \wedge \neg B$, we immediately have $\config{\while{B}{C}}{\sigma'} \exc^* \config{\cskip}{\sigma'}$.

\proofcase{\iwhilerule} This rule is sound for $\PRHL$ only.  Suppose $\sigma' \models P \wedge \neg B$. Since $\sigma' \models P$, we obtain by validity of the rule premise $\sigma \models P \wedge B$ and two possible subcases.  First, if $\config{C}{\sigma}\diverges$ then, because $\sigma \models B$, we get  $\config{\while{B}{C}}{\sigma} \exc \config{C;\while{B}{C}}{\sigma}\diverges$ and are done. Otherwise, $\config{C}{\sigma} \exc^* \config{\cskip}{\sigma'}$. Using the fact that $\sigma' \not\models B$, the operational semantics then gives us
\[\config{\while{B}{C}}{\sigma} \exc \config{C; \while{B}{C}}{\sigma}
\exc^* \config{\while{B}{C}}{\sigma'}
\exc^* \config{\cskip}{\sigma'}\ .\]
Therefore,  $\ihtriple{P \wedge B}{\while{B}{C}}{P \wedge \neg B}$ is valid as required.

\proofcase{\iwhiletotalrule}
This rule is used in $\TRHL$ only.  We first prove the following general statement: for all states $\sigma_k$, if $\sem{t}\sigma_k = k$ and $\sigma_k \models P$ then $\exists \sigma.\ \sigma \models P$ and $\config{\while{B}{C}}{\sigma} \exc^* \config{\while{B}{C}}{\sigma_k}$. We proceed by complete induction on $k$, inductively assuming the statement for all $k' < k$ and showing it then holds for $k$.  Since $\sigma_k \models P \wedge t = k$ by assumption, we have by validity of the rule premise $\sigma_{k'}$ with $\sigma_{k'} \models P \wedge B \wedge t < k$ and $\config{C}{\sigma_{k'}} \exc^* \sigma_k$. Writing $\sem{t}{\sigma_{k'}}=k'$, we then have $k' < k$, and thus by the induction hypothesis we have $\sigma$ with $\sigma \models P$ and $\config{\while{B}{C}}{\sigma} \exc^* \config{\while{B}{C}}{\sigma_{k'}}$. Thus, because $\sigma_{k'}\models B$, we get $\config{\while{B}{C}}{\sigma} \exc^* \config{\while{B}{C}}{\sigma_k}$ by the operational semantics. This completes the induction.

Now, for the main proof, assume $\sigma' \models P \wedge \neg B$, and write $\sem{t}\sigma' = k$ say.  By the inductive statement above, we have $\sigma$ with $\sigma \models P$ and $\config{\while{B}{C}}{\sigma} \exc^* \config{\while{B}{C}}{\sigma'}$.  Because $\sigma' \not\models B$, we can then extend this execution to $\config{\while{B}{C}}{\sigma} \exc^* \config{\cskip}{\sigma'}$ and so are done.
\end{proof}

\begin{theorem}[Relative completeness]
\label{thm:rhl_complete}
Assuming an oracle for logical entailment between assertions, if \ihtriple{P}{C}{Q} is valid in $\TRHL$ then it is provable.

The equivalent result for $\PRHL$ is presently only conjectured.
\end{theorem}

\begin{example}
We show a reverse Hoare logic proof for the total incorrectness triple
$$\ihtriple{x = x_0}{\code{\while{ x>0}{x := x - 2;}}}{\exists k.\ x = x_0-2k \wedge \neg x>0}$$
where $x_0$ is a logical variable expressing the initial value of $x$ before the loop,
$k$ is a natural number.
This loop has variants $x = x_0-2k$
and termination measure $t \equiv (x_0-x)/2$.
Let \code{C} be \code{\while{ x>0}{x := x - 2;}}. The proof is as follows.
\begin{footnotesize}
\[\begin{prooftree}
\[\[
\justifies
\ihtriple{x= x_0 \land \neg x>0}{\code{C}}{x = x_0 \land \neg x>0} 
\]
\justifies
\ihtriple{x= x_0 \land \neg x>0}{\code{C}}{\exists k.\ x = x_0 -2k \land \neg x>0} 
\]
\[
\[\[\[
\justifies
\ihtriple{ x= x_0 - 2n}{\cskip}{x=x_0-2n} \tikz \node (bud2) {}; 
\]
\justifies
\ihtriple{x'=x_0-2(n-1)  \wedge x'>0 \land x=x'-2}{\cskip}{x = x_0-2n} 
\]
\justifies
\ihtriple{x=x_0-2(n-1) \wedge x>0}{x := x - 2}{x = x_0 - 2n}  
\]
\justifies
\ihtriple{x{=} x_0 \land x{>}0}{\code{C}}{\exists k.\ x {=} x_0-2k \land \neg x>0} \tikz \node (comp2) {};
\]
\justifies
\ihtriple{ x= x_0}{\code{C}}{\exists k.\ x = x_0-2k \land \neg x>0} \tikz \node (comp2) {};
\end{prooftree}
\]
\end{footnotesize}

\end{example}

\section{Cyclic proofs in Hoare logic}
\label{sec:cyclic}
This section presents a system of cyclic proofs for $\PHL$ and $\THL$, i.e. partial and total Hoare logic, together with their soundness and their subsumption of the corresponding standard proof systems from the previous section.

First, we give the rules of our cyclic proof system for Hoare triples in Figure~\ref{fig:cyc:correctness}, these being shared for both $\PHL$ and $\THL$.  Compared to their standard equivalents, there are three main points of difference.  First, we formulate the rules in ``continuation style'', where the rule(s) for a program construct $C$ is presented as a rule whose conclusion has the general form $\htriple{P}{C ; C'}{Q}$.  Second, the partial and total rules for {\tt while} loops are replaced with a single rule $\cycwhilerule$ that simply unfolds the loop once (on the left).  Third, we include a rule $\cycsubstrule$ for explicit substitution of variables by expressions; this is sometimes included anyway in Hoare logic (see e.g.~\cite{Apt-Olderog:19}) but is well known to be necessary or at least useful in general for forming backlinks between judgements that are required to be syntactically identical~\cite{Brotherston:PhD,Brotherston-Gorogiannis-Petersen:12}.

\begin{figure*}
\[\begin{array}{c}
\begin{prooftree}
\justifies
\htriple{P}{\cskip}{P}
\using \cyclskiponerule
\end{prooftree}
\qquad
\begin{prooftree}
\htriple{P}{C}{Q}
\justifies
\htriple{P}{\cskip; C}{Q}
\using  \cyclskiprule
\end{prooftree}
\qquad
\begin{prooftree}
\htriple{P[x'/x] \wedge x=E[x'/x]}{C}{Q}
\justifies
\htriple{P}{x:=E;C}{Q}
\using \cycassignrule
\end{prooftree}
\\[1ex]\\
\begin{prooftree}
P' \models P \quad \htriple{P}{C}{Q} \quad Q \models Q'
\justifies
\htriple{P'}{C}{Q'}
\using \cycconseqrule
\end{prooftree}
\qquad
\begin{prooftree}
\htriple{P \wedge B}{C_1;C'}{Q}
\phantom{ww}
\htriple{P \wedge \neg B}{C_2;C'}{Q}
\justifies
\htriple{P}{\ifelse{B}{C_1}{C_2}; C'}{Q}
\using \cyccondrule
\end{prooftree}
\\[1ex] \\
\begin{prooftree}
\htriple{P}{C}{Q}
\justifies
\htriple{P[t/z]}{C}{Q[t/z]}
\using  \substrule
\end{prooftree}
\qquad
\begin{prooftree}
\htriple{P\wedge \neg B}{C'}{Q}
\phantom{ww}
\htriple{P\wedge B}{C;\while{B}{C};C'}{Q}
\justifies
\htriple{P}{\while{B}{C};C'}{Q}
\using \cycwhilerule
\end{prooftree}
\end{array}\]
\caption{\label{fig:cyc:correctness} Cyclic proof rules for Hoare triples.
$x'$ is fresh in $\cycassignrule$, and $z$ and $t$ are not in $\fv(C)$ in $\substrule$.}
\end{figure*}

Next, we define cyclic \emph{pre-proofs} --- derivation trees with \emph{backlinks} between judgement occurrences --- and the \emph{global soundness conditions} qualifying such structures as either partial or total cyclic proofs.

\begin{definition}[Pre-proof]
\label{defn:pre_proof}
A \emph{(cyclic) pre-proof} is a pair $\mathcal{P} = (\mathcal{D,L})$, where $\mathcal{D}$ is a finite derivation tree constructed according to the proof rules and $\mathcal{L}$ is a \emph{back-link function} assigning to every leaf of $\mathcal{D}$ to which no proof rule has been applied (called a {\em bud})
another node of $\mathcal{D}$ (called its \emph{companion}) labelled by an identical proof judgement.
A \emph{leaf} of $\mathcal{D}$ is called \emph{open} if it is not applied with any proof rule.
\end{definition}

We observe that a pre-proof $\mathcal{P}$ can be viewed as a representation of a regular, infinite derivation tree~\cite{Brotherston:PhD}.

The global soundness condition qualifying pre-proofs $\mathcal{P}$ as genuine cyclic proofs is very simple in the case of partial Hoare logic: We simply require that a symbolic execution rule (i.e. not $\cycconseqrule$ or $\cycsubstrule$) is applied infinitely often along every infinite path in $\mathcal{P}$.  Essentially, this guarantees that any putative counterexample to soundness corresponds to an infinite execution of the program and therefore cannot be a counterexample after all.  For total Hoare logic, the situation is a little more complex (cf.~\cite{Brotherston-Bornat-Calcagno:08}); we require that every infinite path in $\mathcal{P}$ contains in the preconditions along the path a \emph{trace} of well-founded measures that decrease, or ``progress'', infinitely often.  This ensures that any putative counterexample cannot after all be terminating and therefore is not a counterexample either.  To formulate traces, since we are using expressions interpreted as natural numbers, we adopt Simpson's definition from \emph{cyclic arithmetic}~\cite{Simpson:17}.

\begin{definition}[Trace~\cite{Simpson:17}]
Let $\mathcal{P} = (\mathcal{D,L})$ be a pre-proof and $(\htriple{P_k}{C_k}{Q_k})_{k{\geq}0}$ be a path in $\mathcal{P}$.
For terms $n$ and $n'$, we say that $n'$ is a \emph{precursor} of $n$ at $k$
if one of the following holds:
\begin{itemize}
\item $\htriple{P_k}{C_k}{Q_k}$ is the conclusion of an application of $\cycassignrule$ or $\cycsubstrule$, and $n' = \theta(n)$ where $\theta$ is the substitution used in the rule application; or
\item $\htriple{P_k}{C_k}{Q_k}$ is the conclusion of another rule, and $n' = n$.
\end{itemize}

A \emph{trace} following $(\htriple{P_k}{C_k}{Q_k})_{k{\geq}0}$ is a sequence of terms $(\trace_k)_{k{\geq}0}$ such that for every $k \geq 0$, the term $\trace_k$ occurs in $P_k$
and also one of the following conditions holds:
\begin{itemize}
\item either $\trace_{k+1}$ is a precursor of $\trace_k$ at $k$; or
\item there exists $n$ such that $(\trace_{k+1} < n) \in P_k$ and $n$ is a precursor of $\trace_k$ at $k$.
\end{itemize}
When the latter case holds, we say that the trace \emph{progresses} (at $k+1$). An \emph{infinitely progressing trace} is a trace that progresses at infinitely many points.
\end{definition}

\begin{definition}[Cyclic proof]
\label{defn:cyclic_proof_HL}
A pre-proof $\mathcal{P}$ is a \emph{cyclic proof} in $\PHL$ if there are infinitely many applications of symbolic execution rules along every infinite path in $\mathcal{P}$. It is furthermore a cyclic proof in $\THL$ if in addition there is an infinitely progressing trace along a tail of every path in $\mathcal{P}$.
\end{definition}

We remark that, by construction, cyclic provability in $\THL$ immediately implies cyclic provability in $\PHL$ as well. We now show two examples, of partial and total cyclic proofs respectively.

\begin{example}
\label{ex:cyclic_PHL}
Let $C$ stand for the program $\while{x>0}{x := x - 2;}$. Here we show
a cyclic proof of $\htriple{ x\geq 0 \land x \% 2=0}{C}{x = 0}$ in $\PHL$:

\begin{footnotesize}
\[\begin{prooftree}
\[\[
\justifies
 \htriple{x = 0 }{\cskip}{x=0}  \using $\cyclskiponerule$
 \]
 \justifies
 \htriple{x\geq 0 \land x \% 2=0 \land \neg(x < 0)}{\cskip}{x=0} \using $\cycconseqrule$
 \]
 \quad
 \[
 \[
 \htriple{x\geq 0 \land x \% 2=0}{C}{x=0} \tikz \node (bud2) {};
 \justifies
 \htriple{x'\geq 0 \land x' \% 2 =0 \land x'>0 \land x=x'-2}{C}{x{=}0} \using $\cycconseqrule$
 \]
 \justifies
 \htriple{x\geq 0 \land x \% 2 =0 \land x>0}{\code{x := x - 2;}C}{x{=}0}  \using $\cycassignrule$
 \]
\justifies
\htriple{x\geq 0 \land x \% 2 = 0}{C}{x = 0} \tikz \node (comp2) {};
\using $\cycwhilerule$
\end{prooftree}
\begin{tikzpicture}[overlay]
\tikzstyle{every path}+=[thick, rounded corners=0.5cm,red,dashed]
 \draw (bud2.north east) -- ++(2.8,0) |- (comp2.north east) [->];
\end{tikzpicture}
\]
\end{footnotesize}
\end{example}

\begin{example}
\label{ex:cyclic_THL}
Continuing with the program $C$ from the previous example, we show a total cyclic proof of
$\htriple{x = 2n}{C}{x = 0}$.

\begin{footnotesize}
\[\begin{prooftree}
\[\[
\justifies
 \htriple{x = 0 }{\cskip}{x=0}  \using $\cyclskiponerule$
 \]
 \justifies
 \htriple{x = 2n \land \neg(x > 0)}{\cskip}{x=0}  \using $\cycconseqrule$
 \] \qquad
 \[
 \[
 \[
 \htriple{\underline{x=2n}}{C}{x=0} \tikz \node (bud2) {};
 \justifies
 \htriple{\underline{x=2(n-1)}}{C}{x=0} \using $\substrule$
 \]
  \justifies
 \htriple{x'=2{n} \land x'>0 \land \underline{x=x'-2}}{C}{x{=}0} \using $\cycconseqrule$
 \]
 \justifies
 \htriple{\underline{x = 2n} \land x>0}{\code{x:=x-2}; C}{x{=}0}  \using $\cycassignrule$
 \]
\justifies
\htriple{\underline{x = 2n}}{C}{x = 0}  \using $\cycwhilerule$ \tikz \node (comp2) {};
\end{prooftree}
\begin{tikzpicture}[overlay]
\tikzstyle{every path}+=[thick, rounded corners=0.5cm,red,dashed]
 \draw (bud2.north east) -- ++(3.2,0) |- (comp2.north east) [->];
\end{tikzpicture}
\]
\end{footnotesize}
In this proof, the progressing trace from companion to bud is given by the \underline{underlined} terms involving $n$.
\end{example}

\begin{lemma}
\label{lem:cyclic_local_sound}
Let $\mathcal{P}$  be a pre-proof of $\htriple{P}{C}{Q}$ and suppose that $\htriple{P}{C}{Q}$ is invalid. Then there exists an infinite path $(\htriple{P_k}{C_k}{Q_k})_{k{\geq}0}$ in $\mathcal{P}$, beginning from $\htriple{P}{C}{Q}$, such that the following properties hold, for all $k \geq 0$:
\begin{enumerate}
\item for all $k \geq 0$, the triple $\htriple{P_k}{C_k}{Q_k}$ is invalid, meaning that either
\begin{enumerate}
\item $\exists \sigma,\sigma'.\ \sigma \models P_k$ and $\config{C_k}{\sigma} \exc^{m_k} \config{\cskip}{\sigma'}$ but $\sigma' \not\models Q_k$; or
\item (in $\THL$ only) $\exists \sigma.\ \sigma \models P_k$ but $\config{C_k}{\sigma}\diverges$.
\end{enumerate}
\item If the first possibility (i)(a) above holds, then the computation length $m_{k+1} \leq m_k$, and if the rule applied at $k$ is a symbolic execution rule then $m_{k+1} < m_k$.
\item (in $\THL$ only) if the second possibility (i)(b) above instead holds, and there is a trace $(\trace_k)_{k \geq i}$ following a tail of the path $(\htriple{P_k}{C_k}{Q_k})_{k \geq i}$, then the sequence of natural numbers defined by $(\sigma(\trace(k)))_{k \geq i}$ is monotonically decreasing, and strictly decreases at every progress point of the trace.
\end{enumerate}
\end{lemma}

\begin{proof}(Sketch)
We construct the required path and prove its needed properties inductively, by analysis of each proof rule in Figure~\ref{fig:cyc:correctness}.
\end{proof}

\begin{theorem}[Soundness]
\label{thm:cyclic_HL_sound}
If $\htriple{P}{C}{Q}$ has a cyclic proof in $\PHL$ (resp. $\THL$) then it is valid in $\PHL$ (resp. $\THL$).
\end{theorem}

\begin{proof}
Suppose first that $\htriple{P}{C}{Q}$ has a cyclic proof $\mathcal{P}$ in $\PHL$ but is invalid there. By Lemma~\ref{lem:cyclic_local_sound} we can build an infinite path $(\htriple{P_k}{C_k}{Q_k})_{k{\geq}0}$ in $\mathcal{P}$, beginning from $\htriple{P}{C}{Q}$, satisfying properties (i)(a) and (ii) above.  In particular, we have an infinite, monotonically decreasing sequence $(m_k)_{k \geq 0}$ of natural numbers such that for all $k \geq 0$ we have $\config{C_k}{\sigma} \exc^{m_k} \config{\cskip}{\sigma'}$ (for some $\sigma$ and $\sigma'$).  Since $\mathcal{P}$ is a cyclic proof in $\PHL$, this path contains infinitely many applications of symbolic execution rules, and thus by property (ii) the sequence $(m_k)_{k \geq 0}$, which is a contradiction.

Next suppose that $\mathcal{P}$ is also a cyclic proof in $\THL$ but is not valid in $\THL$.  We apply Lemma~\ref{lem:cyclic_local_sound} to obtain an infinite path in $\mathcal{P}$ as above, satisfying properties (i)(b) and (iii); since it is already a cyclic proof in $\PHL$, we may rule out $(i)(a)$ as a possibility (and need not consider (ii) either). Since $\mathcal{P}$ is a cyclic proof in $\THL$, there is an infinitely progressing trace $(\trace_k)_{k \geq i}$ following some tail ($k \geq i$) of this path.  Thus by property (iii) there is an infinite, monotonically decreasing sequence $(\sigma(\trace_k))_{k \geq i}$ of numbers that moreover strictly decreases infinitely often; again a contradiction.  Thus $\htriple{P}{C}{Q}$ is valid in $\THL$ after all.
\end{proof}

\begin{lemma}[Proof translation]
\label{lem:HL_translation}
If $\htriple{P}{C}{Q}$ is provable in $\PHL$ (resp. $\THL$), then for all statements $C'$ and assertions $R$, there is a cyclic pre-proof of $\htriple{P}{C;C'}{R}$ in $\PHL$ (resp. $\THL$) in which all open leaves are occurrences of $\htriple{Q}{C'}{R}$:
\[\begin{array}{c@{\hspace{1cm}}c@{\hspace{1cm}}c}
\begin{prooftree}
\leadsto
\htriple{P}{C}{Q}
\end{prooftree}
& \Longrightarrow &
\begin{prooftree}
\htriple{Q}{C'}{R}
\leadsto
\htriple{P}{C;C'}{R}
\end{prooftree}
\end{array}\]
Moreover, any strongly connected subgraph of the pre-proof created by the translation satisfies the global soundness condition for $\PHL$ (resp. $\THL$) cyclic proofs.
\end{lemma}

\begin{proof}
We proceed by structural induction on the Hoare logic proof of $\htriple{P}{C}{Q}$,
distinguishing cases on the last rule applied in the proof and assuming arbitrary $C'$ and $R$.

\proofcase{\cskiprule} The proof transformation is as follows:
\[\begin{array}{c@{\hspace{1cm}}c@{\hspace{1cm}}c}
\begin{prooftree}
\phantom{P}
\justifies
\htriple{P}{\cskip}{P}
\using \cskiprule
\end{prooftree}
& \Longrightarrow &
\begin{prooftree}
\htriple{P}{C'}{R}
\justifies
\htriple{P}{\cskip;C'}{R} \using \cyclskiprule
\end{prooftree}
\end{array}\]
The only open leaf in the cyclic pre-proof is an instance of $\htriple{P}{C'}{R}$, as required.

\proofcase{\cassignrule}
\[\begin{array}{c@{\hspace{1cm}}c@{\hspace{1cm}}c}
\begin{prooftree}
\phantom{P}
\justifies
\htriple{P}{x := E}{P[x'/x] \land x = E[x'/x]}
\using \cassignrule
\end{prooftree}
& \Longrightarrow &
\begin{prooftree}
\htriple{P[x'/x] \wedge x = E[x'/x]}{C'}{R}
\justifies
\htriple{P}{x:=E; C'}{R}
\using \cycassignrule
\end{prooftree}
\end{array}\]
The only open leaf in this pre-proof is an instance of $\htriple{P[x'/x] \wedge x = E[x'/x]}{C'}{R}$, as required.

\proofcase{\cconseqrule}
\[\begin{array}{c@{\hspace{1cm}}c@{\hspace{1cm}}c}
\begin{prooftree}
\[ \leadsto
P \models P' \quad \htriple{P'}{C}{Q'} \quad Q' \models Q \]
\justifies
\htriple{P}{C}{Q}
\using \cconseqrule
\end{prooftree}
& \Longrightarrow &
\begin{prooftree}
P \models P' \quad
\[\[\htriple{Q}{C}{R}
\justifies
\htriple{Q'}{C'}{R} \using \cycconseqrule\]
\leadsto
\htriple{P'}{C;C'}{R} \using \;\indhyp \]
\justifies
\htriple{P}{C;C'}{R} \using \cycconseqrule
\end{prooftree}
\end{array}\]
On the RHS we first use the consequence rule to transform the precondition $P$ to $P'$, then apply the induction hypothesis (marked $\indhyp$ in the derivation) to obtain a cyclic pre-proof with open leaves of the form $\htriple{Q'}{C'}{R}$.  By applying the consequence rule to each of these open leaves we obtain a pre-proof with open leaves of the form $\htriple{Q}{C'}{R}$ as required.

\proofcase{\csequencerule}
\[\begin{array}{c@{\hspace{1cm}}c@{\hspace{1cm}}c}
\begin{prooftree}
\[\leadsto
\htriple{P}{C_1}{S} \]
\quad
\[\leadsto
\htriple{S}{C_2}{Q}\]
\justifies
\htriple{P}{C_1;C_2}{Q}
\using C
\end{prooftree}
& \Longrightarrow &
\begin{prooftree}
\[\htriple{Q}{C'}{R}
\leadsto
\htriple{S}{C_2 ; C'}{R} \using \;\indhyp \]
\leadsto
\shiftright 1.2em \htriple{P}{C_1;C_2;C'}{R} \using \;\indhyp
\end{prooftree}
\end{array}\]
Here, we first use the induction hypothesis with the first premise $\htriple{P}{C_1}{S}$ of $\htriple{P}{C_1}{S}$ to yield a pre-proof of $\htriple{P}{C_1;C_2;C'}{R}$ with open leaves all of form $\htriple{S}{C_2 ; C'}{R}$.  Then, we use the induction hypothesis with the second premise to expand  these leaves into pre-proofs with open leaves all of form $\htriple{Q}{C'}{R}$, as needed.

\proofcase{\ccondrule}
\[\begin{array}{c@{\hspace{0.7cm}}c@{\hspace{0.7cm}}c}
\begin{prooftree}
\[ \leadsto
\htriple{P \wedge B}{C_1}{Q} \]
\[ \leadsto
\htriple{P \wedge \neg B}{C_2}{Q} \]
\justifies
\htriple{P}{\ifelse{B}{C_1}{C_2}}{Q}
\using \ccondrule
\end{prooftree}
& \Longrightarrow &
\begin{prooftree}
\[\htriple{Q}{C'}{R}
\leadsto
\htriple{P \wedge B}{C_1;C'}{R} \using \;\indhyp \]
\quad
\[\htriple{Q}{C'}{R}
\leadsto
\htriple{P \wedge \neg B}{C_2;C'}{R} \using \;\indhyp \]
\vspace{0.1mm}
\justifies
\htriple{P}{\ifelse{B}{C_1}{C_2};C'}{R} \using \cyccondrule
\end{prooftree}
\end{array}\]

\proofcase{\cwhilerule} Here the proof transformation involves creating (possibly many) new backlinks:
\[\begin{array}{c@{\hspace{0.3cm}}c@{\hspace{0.3cm}}c}
\begin{prooftree}
\[\leadsto
\htriple{P\wedge B}{C}{P} \]
\justifies
\htriple{P}{\while{B}{C}}{P \land \neg B}
\using \cwhilerule
\end{prooftree}
& \Longrightarrow &
\begin{prooftree}
\htriple{P \land \neg B}{C'}{R}
\;
\[\htriple{P}{\while{B}{C; C'}}{R} \tikz \node (bud5) {};
\leadsto
\htriple{P \land B}{C; \while{B}{C; C'}}{R} \using \;\indhyp \]
\vspace{0.1mm}
\justifies
\htriple{P}{\while{B}{C; C'}}{R} \tikz \node (comp5) {};
\using \cycwhilerule
\begin{tikzpicture}[overlay]
\tikzstyle{every path}+=[thick, rounded corners=0.5cm,red,dashed]
\draw (bud5.north east) -- ++(2.1,0) |- (comp5.north east) [->];
\end{tikzpicture}
\end{prooftree}
\end{array}\]
In the RHS pre-proof, we first apply the cyclic rule $\cycwhilerule$ to unfold the {\tt while} loop.  The resulting left-hand premise is an open leaf of the permitted form, i.e. $\htriple{P \wedge \neg B}{C'}{R}$.  For the right-hand premise, using the induction hypothesis we can obtain a pre-proof of $\htriple{P \land B}{C; \while{B}{C; C'}}{R}$ with all open leaves of the form $\htriple{P}{\while{B}{C; C'}}{R}$.  These leaves are all back-linked to the conclusion of the pre-proof, which is identical.  We additionally note that at least one symbolic execution rule is applied along the path from this companion to any of these buds, namely the instance of $\cycwhilerule$ itself.

\proofcase{\cwhiletotalrule} By assumption, we have a proof of the form:
\[\begin{prooftree}
\[\leadsto
\htriple{P \wedge B \land t = n}{C}{P \land t<n} \]
\justifies
\htriple{P}{\while{B}{C}}{P \wedge \neg B}
\using \cwhiletotalrule
\end{prooftree}\]
We derive a cyclic pre-proof of $\htriple{P}{\while{B}{C; C'}}{R}$ as follows.
\[\begin{prooftree}
\[
\[\htriple{P \land \neg B}{C'}{R}
\justifies
\htriple{P \land \neg B \wedge t = n}{C'}{R} \using \cycconseqrule \]
\;
\[\[\[\[
\htriple{P \wedge t=n}{\while{B}{C; C'}}{R} \tikz \node (bud5) {};
\justifies
\htriple{P \wedge t=n'}{\while{B}{C; C'}}{R} \using \cycsubstrule \]
\justifies
\htriple{P \wedge t=n' \wedge n' < n}{\while{B}{C; C'}}{R} \using \cycconseqrule \]
\justifies
\htriple{P \wedge t<n}{\while{B}{C; C'}}{R} \using \cycconseqrule \]
\leadsto
\htriple{P \wedge B \wedge t = n}{C; \while{B}{C; C'}}{R} \using \;\indhyp \]
\justifies
\htriple{P \wedge t = n}{\while{B}{C; C'}}{R} \tikz \node (comp5) {};
\using \cycwhilerule \]
\justifies
\htriple{P}{\while{B}{C; C'}}{R} \using \cycconseqrule
\begin{tikzpicture}[overlay]
\tikzstyle{every path}+=[thick, rounded corners=0.5cm,red,dashed]
\draw (bud5.north east) -- ++(2.6,0) |- (comp5.north east) [->];
\end{tikzpicture}
\end{prooftree}\]
This construction is similar to the previous case, with a little more wrangling to deal with the termination measure $t$.  First, before unfolding the {\tt while} loop we ``record'' the value of $t$ as a fresh variable $n$ to obtain $t=n$ in the precondition.  In the left hand premise of $\cycwhilerule$ this fact is not needed and is discarded again to obtain an open leaf of the permitted form $\htriple{P \land \neg B}{C'}{R}$.  In the right-hand premise, we apply the induction hypothesis to obtain a pre-proof with open leaves all of form $\htriple{P \wedge t<n}{\while{B}{C; C'}}{R}$.  In each of these open leaves we introduce another fresh variable $n'$ to record the new value of $t$ as $t=n'$, where $n'<n$, thus recognising these transformed leaves as substitution instances of the conclusion of $\cycwhilerule$, to which we form backlinks.

Note that we have a trace on $n$ and $n'$ from the companion node in this proof to each of the buds, which progresses when we ``jump'' from $n$ to $n'$ (at the point where $n' < n$ is introduced).
\end{proof}

\begin{theorem}[Relative completeness]
\label{thm:HL_translation}
If $\htriple{P}{C}{Q}$ is provable in $\PHL$ (resp. $\THL$) then it has a cyclic proof in $\PHL$ (resp. $\THL$).
\end{theorem}

\begin{proof}
By taking $C'=\cskip$ and $R=Q$ in Lemma~\ref{lem:HL_translation} and using the elision of $C; \cskip$ to $C$, we obtain a pre-proof $\mathcal{P}$ of $\htriple{P}{C}{Q}$ in $\PHL$ (resp. $\THL$) in which all open leaves are of the form $\htriple{P}{\cskip}{P}$ and thus can be immediately closed by applications of the rule $\cyclskiponerule$.

To see that $\mathcal{P}$ is a genuine cyclic proof, observe that the only strongly connected subgraphs (i.e. collections of cycles) in $\mathcal{P}$ are created by translations of $\cwhilerule$ (in $\PHL$) and $\cwhiletotalrule$ (in $\THL$); these are disjoint from one another by construction, and satisfy the global soundness conditions for $\PHL$ and $\THL$ respectively.  Thus $\mathcal{P}$ itself is a cyclic proof in $\PHL$ / $\THL$ as required.
\end{proof}

\section{Cyclic proofs in reverse Hoare logic}
\label{sec:cyclic:il:hoare}
This section presents a system of cyclic proofs for $\PRHL$ and $\TRHL$, i.e. partial and total reverse Hoare logic, together with their soundness and their subsumption of the corresponding standard proof systems from Section~\ref{sec.two.hoare}. By and large, their development mirrors that of the cyclic proof systems for standard Hoare logic in the previous section.

First, we give our cyclic proof rules for reverse Hoare triples
in Figure~\ref{fig:cyclic_incorrectness_rules}. Like the cyclic proof rules for standard Hoare triples in the previous section, they are formulated in ``continuation style'' (with conclusions of general form $\ihtriple{P}{C;C'}{Q}$, with a rule for explicit substitution $ \isubstrule$ and an unfolding rule for {\tt while} loops $\icycwhilerule$.

\begin{figure}[tb]
\[\begin{array}{c}
\begin{prooftree}
\justifies
\ihtriple{P}{\cskip}{P}
\using \crskiponerule
\end{prooftree}
\qquad
\begin{prooftree}
\ihtriple{P}{C}{Q}
\justifies
\ihtriple{P}{\cskip; C}{Q}
\using \crskiprule
\end{prooftree}
\qquad
\begin{prooftree}
\ihtriple{P[x'/x] \wedge x=E[x'/x]}{C}{Q}
\justifies
\ihtriple{P}{x:=E; C}{Q}
\using \icycassignrule
\end{prooftree}
\\[1ex] \\
\begin{prooftree}
\ihtriple{P}{C}{Q}
\justifies
\ihtriple{P[t/z]}{C}{Q[t/z]}
\using  \isubstrule
\end{prooftree}
\qquad
\begin{prooftree}
P \models P' \quad \ihtriple{P}{C}{Q} \quad
Q' \models Q
\justifies
\ihtriple{P'}{C}{Q'}
\using \crconseqrule
\end{prooftree}
\qquad
\begin{prooftree}
  \ihtriple{P_1}{C}{Q_1} \quad
\ihtriple{P_2}{C}{Q_2}
\justifies
\ihtriple{P_1 \lor P_2}{C}{Q_1 \lor Q_2}
\using \idisjtrule
\end{prooftree}
\\[1ex] \\
\begin{prooftree}
\ihtriple{P \wedge B}{C_1; C'}{Q}
\justifies
\ihtriple{P \wedge B}{\ifelse{B}{C_1}{C_2}; C'}{Q}
\using \icyccondtrule
\end{prooftree}
\qquad
\begin{prooftree}
\ihtriple{P \wedge \neg B}{C_2; C'}{Q}
\justifies
\ihtriple{P \wedge \neg B}{\ifelse{B}{C_1}{C_2} ; C'}{Q}
\using \icyccondfrule
\end{prooftree}
\\ [1ex] \\
\begin{prooftree}
\ihtriple{P \wedge \neg B}{C'}{Q}
\justifies
\ihtriple{P}{\while{B}{C}; C'}{Q}
\using \icycwhileskiprule
\end{prooftree}
\qquad
\begin{prooftree}
\ihtriple{P \wedge B}{C ; \while{B}{C} ; C'}{Q}
\justifies
\ihtriple{P}{\while{B}{C} ; C'}{Q}
\using \icycwhilerule
\end{prooftree}
\end{array}\]
\caption{\label{fig:cyclic_incorrectness_rules} Cyclic proof rules for reverse Hoare triples.
$x'$ is fresh in $\cycassignrule$, and $z$ and $t$ are not in $\fv(C)$ in $\isubstrule$.}
\end{figure}

\begin{definition}[Trace]
Let $\mathcal{P}$ be a pre-proof and $(\ihtriple{P_k}{C_k}{Q_k})_{k{\geq}0}$ be a path in $\mathcal{P}$.
For terms $n$ and $n'$, we say that $n'$ is a \emph{precursor} of $n$ at $k$ if one of the following holds:
\begin{itemize}
\item either $\ihtriple{P_k}{C_k}{Q_k}$ is the conclusion of an application of $\isubstrule$ 
and $n' = \theta(n)$, where $\theta$ is the substitution used in the rule application; or
\item $\ihtriple{P_i}{C_i}{Q_i}$ is the conclusion of another rule, and $n' = n$.
\end{itemize}

A \emph{trace} following $(\ihtriple{P_k}{C_k}{Q_k})_{k{\geq}0}$ is a sequence of terms $(\trace_k)_{k{\geq}0}$ such that for every $k \geq 0$, the term $\trace_k$ occurs in either $P_k$ or $Q_k$ and one of the following conditions holds:
\begin{itemize}
\item either $\trace_{i+1}$ is a precursor of $\trace_k$ at $k$; or
\item there exists $(t = \trace_{k+1}) \in P_i$ such that $\trace_k < \trace_{k+1}$
where $\trace_k$ is a precursor
of $\trace_{k+1}$ at $i$.
\end{itemize}
In the latter case, we say that the trace \emph{progresses} at $k+1$. An \emph{infinitely progressing trace} is a trace that progresses at infinitely many points.
\end{definition}

A pre-proof is a genuine cyclic proof if it satisfies the following global soundness condition(s).

\begin{definition}[Cyclic proof]
\label{defn:cyclic_proof_RHL}
A pre-proof $\mathcal{P}$ is a \emph{cyclic proof} in $\PRHL$ if there are infinitely
many symbolic execution rule applications along every infinite path in $\mathcal{P}$.
If in addition there is an infinitely progressing trace along a tail of every infinite path in $\mathcal{P}$, it is also a cyclic proof in $\TRHL$.
\end{definition}

Similarly to Definition~\ref{defn:cyclic_proof_HL}, every cyclic proof in $\TRHL$ is also a cyclic proof in $\PRHL$.

\begin{example}\label{ex:cyclic_PRHL}
We show a cyclic proof in $\TRHL$ for  the reverse Hoare triple
$$\ihtriple{x = x_0-2n}{\code{\while{x>0} {x := x - 2;}}}{\exists k.\ x =x_0-2k \wedge \neg x >{0}}$$
In the following cyclic proof, we use  $\code{C}$ to denote the program $\code{\while{x>0} {x := x - 2;}}$ and
 $Q$ for the post-condition $\exists k.\ x =x_0-2k \wedge \neg x >{0}$.
\begin{footnotesize}
\[\begin{prooftree}
\[\[\[
\justifies
\ihtriple{x =x_0-2n \wedge \neg x >{0}}{\cskip}{x =x_0-2n \wedge \neg x >{0}} \using
 \crskiponerule
 \]
 \justifies
\ihtriple{x =x_0-2n \wedge \neg x >{0}}{\cskip}{Q} \using
 \iconseqrule
 \]
 \justifies
\ihtriple{x =x_0-2n \wedge \neg x >{0}}{\code{C}}{Q} \using
 \icycwhileskiprule
\]
 \[\[ \[\[
 \ihtriple{\underline{x=x_0-2n}}{\code{C}}{Q}  \tikz \node (bud2) {};
    \justifies
 \ihtriple{\underline{x=x_0-2(n+1)}}{\code{C}}{Q}  \using \isubstrule 
 \]
 \justifies
 \ihtriple{x' = x_0-2n \land x'>0 \land \underline{x = x'-2}}{\code{C}}{Q}  \using \iconseqrule
 \]
 \justifies
 \ihtriple{\underline{x = x_0-2n} \land x>0}{\code{x := x - 2; C}}{Q}  \using  \icycassignrule
 \]
 \justifies
 \ihtriple{\underline{x = x_0-2n \wedge x>0}}{\code{C}}{Q}  \using \icycwhilerule
 \]
\justifies
\ihtriple{\underline{x = x_0-2n}}{\code{\while{x>0} {x := x - 2;}}}{\exists k.\ x =x_0-2k \wedge \neg x >{0}} \tikz \node (comp2) {};
\using \idisjtrule
\end{prooftree}
\begin{tikzpicture}[overlay]
\tikzstyle{every path}+=[thick, rounded corners=0.5cm,red,dashed]
 \draw (bud2.north east) -- ++(3.3,0) |- (comp2.north east) [->];
\end{tikzpicture}
\]
\end{footnotesize}
Note that, in this proof, the trace from companion to bud is \underline{underlined}
and the progressing point is at the instance of rule $\isubstrule$.
It can be confirmed
that the induced infinite path includes infinitely instances of symbolic application and infinite
progressing points.

\end{example}

\begin{lemma}
\label{lem_cyclic_incorrectness_invalid}
Let $\mathcal{P}$  be a pre-proof of $\ihtriple{P}{C}{Q}$ and suppose that $\ihtriple{P}{C}{Q}$ is invalid. Then there exists an infinite path $(\ihtriple{P_k}{C_k}{Q_k})_{k{\geq}0}$ in $\mathcal{P}$, beginning from $\ihtriple{P}{C}{Q}$, such that the following properties hold, for all $k \geq 0$:
\begin{enumerate}
\item For all $k \geq 0$, the triple $\ihtriple{P_k}{C_k}{Q_k}$ is invalid, meaning that
$\exists \sigma'.\ \sigma' \models Q_k$, there does not exist $\sigma.\ \sigma \models P_k$ and $\config{C_k}{\sigma} \exc^{m_k} \config{\cskip}{\sigma'}$
and the computation length $m_{k+1} \leq m_k$, and if the rule applied at $k$ is a symbolic execution rule then $m_{k+1} < m_k$; or 
 
 \item (in $\TRHL$ only)  If 
 there is a trace $(\trace_k)_{k \geq i}$ following a tail of the path $(\ihtriple{P_k}{C_k}{Q_k})_{k \geq i}$, then the sequence of natural numbers defined by $(\sigma(\trace(k)))_{k \geq i}$ is monotonically decreasing, and strictly decreases at every progress point of the trace.  
\end{enumerate}
\end{lemma}

\begin{proof}
By structural induction on the rules in Fig. \ref{fig:cyclic_incorrectness_rules}
and for each rule while we prove (i) by contradiction, (ii) can be derived from
the definition of progressing trace.
Here for (i), we just show the more interesting cases.

\proofcase{\icycassignrule}
Supposing $\sigma' \models Q$ such that there does not exist $\sigma$ and ${m_k}$,
$\sigma \models P$ and $\config{x:=E;C}{\sigma} \exc^{m_k} \config{\cskip}{\sigma'}$.

We assume there exists $\sigma_1$ such that  $\sigma_1 \models P[x'/x] \wedge x=E[x'/x]$, and
$\config{C}{\sigma_1} \exc^{m} \config{\cskip}{\sigma'}$ (1).

Then, from $\sigma_1$, we construct $\sigma$ as follows. $\dom(\sigma) = \dom(\sigma_1) \setminus\{x'\}$, for all $y \neq x.\ \sigma(y)=\sigma_1(y)$, $\sigma(x) = \sigma_1(x')$.
Then, $\sigma \models P$ and the operational semantics gives us $\config{x:=E}{\sigma} \exc^1 \config{\cskip}{\sigma_1}$.

Therefore, there exists $\sigma \models P$, $\config{x:=E;C}{\sigma} \exc^{1} \config{\cskip;C}{\sigma_1}$. Using the operational semantics, we can obtain
$\config{x:=E;C}{\sigma} \exc^{1} \config{C}{\sigma_1}$. Then combining with
(1), $\config{x:=E;C}{\sigma} \exc^{m+1} \config{\cskip}{\sigma'}$. Contradiction.

\proofcase{\icyccondtrule}
Supposing $\sigma' \models Q$ such that there does not exist $\sigma$ and ${m_k}$,
$\sigma \models P$ and

$\config{\ifelse{B}{C_1}{C_2}; C'}{\sigma} \exc^{m_k} \config{\cskip}{\sigma'}$.

\noindent We assume the premise is valid, i.e., there exists $\sigma$ s.t.  $\sigma \models P \wedge B$, and
$\config{C_1;C'}{\sigma} \exc^{m} \config{\cskip}{\sigma'}$ (2).

From (2) we can imply  $\sigma \models B$, and then
the operational semantics gives us
$\config{\ifelse{B}{C_1}{C_2}}{\sigma} \exc^1 \config{C_1}{\sigma}$
and then $\config{\ifelse{B}{C_1}{C_2}; C'}{\sigma} \exc^1 \config{C_1; C'}{\sigma}$.
We now combine this execution with (2), and obtain
$\config{\ifelse{B}{C_1}{C_2}; C'}{\sigma} \exc^{m+1} \config{\cskip}{\sigma'}$.
Contradiction.

\proofcase{\icycwhilerule}
Supposing $\sigma' \models Q$ such that there does not exist $\sigma$ and ${m_k}$,
$\sigma \models P$ and

$\config{\while{B}{C} ; C'}{\sigma} \exc^{m_k} \config{\cskip}{\sigma'}$.

We assume the premise is valid, i.e., there exists $\sigma$ s.t.  $\sigma \models P \wedge B$, and
$\config{C ; \while{B}{C} ; C'}{\sigma} \exc^{m} \config{\cskip}{\sigma'}$ (3).

From (3) we can imply  $\sigma \models B$, and then
the operational semantics gives us
$\config{\while{B}{C}}{\sigma} \exc^1 \config{C; \while{B}{C}}{\sigma}$
and then $\config{\while{B}{C}; C'}{\sigma} \exc^1 \config{C; \while{B}{C}; C'}{\sigma}$.
We now combine this execution with (3), and obtain
$\config{\while{B}{C}; C'}{\sigma} \exc^{m+1} \config{\cskip}{\sigma'}$.
Contradiction.
\end{proof}

\begin{theorem}[Soundness]
\label{thm:cyclic_RHL_soundness}
If $\ihtriple{P}{C}{Q}$ has a cyclic proof in $\PRHL$ (resp. $\TRHL$) then it is valid in $\PRHL$ (resp. $\TRHL$).
\end{theorem}

\begin{proof}
Similar to Theorem~\ref{thm:cyclic_HL_sound}, using Lemma~\ref{lem_cyclic_incorrectness_invalid}.
\end{proof}

\begin{lemma}[Proof translation]
\label{lem:RHL_translation}
If $\ihtriple{P}{C}{Q}$ is provable in $\PRHL$ (resp. $\TRHL$) then, for all statements $C'$ and assertions $R$, there is a cyclic pre-proof of $\ihtriple{P}{C; C'}{R}$ in which all open leaves are occurrences of $\ihtriple{Q}{C'}{R}$:
\[\begin{array}{c@{\hspace{1cm}}c@{\hspace{1cm}}c}
\begin{prooftree}
\leadsto
\ihtriple{P}{C}{Q}
\end{prooftree}
& \Longrightarrow &
\begin{prooftree}
\ihtriple{Q}{C'}{R}
\leadsto
\ihtriple{P}{C;C'}{R}
\end{prooftree}
\end{array}\]
Moreover, any strongly connected subgraph of the pre-proof created by the translation satisfies the global soundness condition for $\PRHL$ (resp. $\TRHL$) cyclic proofs.
\end{lemma}

\begin{proof}
We proceed by structural induction on the Hoare logic proof of $\htriple{P}{C}{Q}$,
distinguishing cases on the last rule applied in the proof and assuming arbitrary $C'$ and $R$. The rules $\iskiprule$ and $\iwhileskiprule$ are trivial, while the cases of $\iassignrule$, $\isequencerule$, $\iconseqrule$ are similar to their counterparts in standard Hoare logic (Lemma~\ref{lem:HL_translation}).

\proofcase{\icondtrule, \icondfrule} These cases are symmetric.  For $\icondtrule$, the transformation is as follows:
\[\begin{array}{c@{\hspace{0.7cm}}c@{\hspace{0.7cm}}c}
\begin{prooftree}
\[\leadsto
\ihtriple{P \wedge B}{C_1}{Q}\]
\justifies
\ihtriple{P \wedge B}{\ifelse{B}{C_1}{C_2}}{Q}
\using \icondtrule
\end{prooftree}
& \Longrightarrow &
\begin{prooftree}
\[
\ihtriple{Q}{C}{R}
\leadsto
\ihtriple{P\wedge B}{C_1; C}{R}
  \using \;\indhyp
\]
\justifies
\ihtriple{P\wedge B}{\ifelse{B}{C_1}{C_2}; C}{Q}
\using \icyccondtrule
\end{prooftree}
\end{array}\]

\proofcase{\iwhilerule}
\[\begin{array}{c@{\hspace{0.7cm}}c@{\hspace{0.7cm}}c}
\begin{prooftree}
\[\leadsto
\ihtriple{P \wedge B}{C}{P} \]
\justifies
\ihtriple{P }{\while{B}{C}}{P \wedge \neg B}
\using \iwhilerule
\end{prooftree}
& \Longrightarrow &
\begin{prooftree}
\[
\ihtriple{P}{\while{B}{C}; C'}{R} \tikz \node (bud5) {}; \using \isubstrule
\leadsto
\ihtriple{P \land B}{C; \while{B}{C}; C'}{R}
\using \;\indhyp \]
\vspace{0.1mm}
\justifies
\ihtriple{P}{\while{B}{C}; C'}{R} \tikz \node (comp5) {};
\using \icycwhilerule
\begin{tikzpicture}[overlay]
\tikzstyle{every path}+=[thick, rounded corners=0.5cm,red,dashed]
\draw (bud5.north east) -- ++(2.7,0) |- (comp5.north east) [->];
\end{tikzpicture}
\end{prooftree}
\end{array}\]
In the pre-proof above, by using the induction hypothesis, a pre-proof of $\ihtriple{P \land B}{C; \while{B}{C}; C'}{R}$ contains open leaves which are occurrences of $\ihtriple{P}{\while{B}{C}; C'}{R}$.
As those occurrences are involved in a back-link, the pre-proof has no open leaves. (In this case, as the first condition in global soundness holds, this pre-proof is a $\PRHL$ cyclic proof.)

\proofcase{\iwhiletotalrule}
By assumption, we have a proof of the form:
    \[\begin{prooftree}
      \[\leadsto
        \ihtriple{P \wedge B \wedge t<n}{C}{P \wedge t=n} \]
        \justifies
        \ihtriple{P \wedge t=n_0}{\while{B}{C}}{P \wedge \neg B}
        \using \cwhiletotalrule
    \end{prooftree}\]
    where $n'$ is a fresh variable.
     We derive a cyclic pre-proof of \ihtriple{P}{\while{B}{C; C'}}{R} as follows.
     \[\begin{prooftree}
     \[\[
     \ihtriple{P \wedge \neg B}{C'}{R}
     \justifies
      \ihtriple{P \wedge t=n \wedge \neg B}{C'}{R}  \using \crconseqrule
      \]
     \justifies
      \ihtriple{P \wedge t=n}{\while{B}{C};C'}{R} \using \icycwhileskiprule 
     \]\quad
     \[
      \[\[
      \ihtriple{P \wedge t = n_0}{\while{B}{C; C'}}{R} \tikz \node (bud5) {};
         \justifies
         \ihtriple{P \wedge t = n}{\while{B}{C; C'}}{R} \using \isubstrule \]
        \leadsto
        \ihtriple{P  \wedge t < n \wedge B }{C; \while{B}{C}; C'}{R}
        \using \;\indhyp \]
        \justifies
        \ihtriple{P \wedge  t  < n}{\while{B}{C}; C'}{R}
        \using \icycwhilerule
        \]
        \justifies
        \ihtriple{P \wedge t=n_0}{\while{B}{C}; C'}{R} \tikz \node (comp5) {};
        \using \idisjtrule
        \begin{tikzpicture}[overlay]
          \tikzstyle{every path}+=[thick, rounded corners=0.5cm,red,dashed]
    \draw (bud5.north east) -- ++(2.6,0) |- (comp5.north east) [->];
    \end{tikzpicture}
    \end{prooftree}\]
  In the pre-proof above, the left hand premise
  includes an open leaf of the permitted form $\ihtriple{P \land \neg B}{C'}{R}$.
  The right hand premise,
  by applying the induction hypothesis,
  $\ihtriple{P \land  t < n \wedge B}{C; \while{B}{C; C'}}{R}$ has
  a cyclic pre-proof
    in which all open leaves are of the form
  $\ihtriple{P \land t = n}{\while{B}{C}; C'}{R}$.
  In each of these open leaves, we introduce a case split at $\idisjtrule$
  with base case $t=n$ and step case $t < n$.
  Similar to the translation in Hoare logic, in this proof
  we have a trace on the values of $t$ from the companion node  to each of the buds, in which
  the progress point is at point when we ``jump'' from $n$ to $n_0 < n$ (rule $\isubstrule$).
  Note that this infinite path also includes infinitely
many symbolic execution rule applications of $\icycwhilerule$.
So, the global conditions of a cyclic proof in $\TRHL$ holds in this infinite path.
\end{proof}

\begin{theorem}[Relative completeness]
\label{thm:RHL_translation}
If $\htriple{P}{C}{Q}$ is provable in $\PRHL$ (resp. $\TRHL$) then it has a cyclic proof in $\PRHL$ (resp. $\TRHL$).
\end{theorem}

\begin{proof}
Similar to Theorem~\ref{thm:HL_translation}, using Lemma~\ref{lem:RHL_translation}.
\end{proof}

\section{Conclusions and future work}
\label{sec:conclusion}

This paper presents formulations of the partial and total versions of Hoare logic and its reverse, semantically and as axiomatic and cyclic proof systems.  We observe in particular that
\begin{itemize}
\item the total versions of the logics are special cases of the corresponding partial version;
\item partial reverse Hoare logic $\PRHL$ and total Hoare logic $\THL$ are semantic duals, as are $\TRHL$ and $\PHL$;
\item the partial versions of the logics are proof-theoretically similar, as are the total versions, with their cyclic proof systems sharing a similar soundness condition;
\item there is a natural translation from standard to cyclic proofs for each of the logics.
\end{itemize}

We must make a frank admission: Very little of what we present here is truly new, in that it could in principle have been distilled from previous works on Hoare logic, reverse Hoare logic and cyclic proof.  Of course, the partial and total variants of Hoare logic have been extensively studied for decades~\cite{Hoare:69,Cook:1978,Apt:CACM:1981,Apt-Olderog:19} and their formulations as cyclic proof systems can largely be inferred from previous works, in particular, on Hoare-style proofs in separation logic~\cite{Brotherston-Bornat-Calcagno:08,Brotherston-Gorogiannis:14}.  Meanwhile, reverse Hoare logic and its extension to incorrectness logic has been studied in~\cite{Edsko:SEFM:2011,OHearn:19,Le:OOPSLA:2022}, both semantically and axiomatically, although not to our knowledge in cyclic proof form.  Moreover, Verscht and Kaminski provide a semantic taxonomy of Hoare-like logics covering many more possibilities than the four we examine here~\cite{Verscht-Kaminski:25}.  We see our main contribution as being primarily one of compiling and formulating these logics in such a way that their proof-theoretic as well as their semantic relationships become clear.

We can additionally mention a few minor novelties of the present work.  First, the partial variant of reverse Hoare logic we consider here, $\PRHL$, does not seem to be well known, apparently not being among the logics in the Verscht-Kaminski taxonomy~\cite{Verscht-Kaminski:25} and may even be somewhat new.  A very recent short paper by Verscht et al~\cite{Verscht-etal:25} describes ``partial incorrectness logic'' which \emph{also} does not seem at first sight the same thing as our $\PRHL$; however, we believe that our version arises quite naturally as a semantic dual of total correctness (Defn.~\ref{defn:valid_triples}) and a proof-theoretic dual of partial correctness (Figure~\ref{fig:ihoare_rules}).  Second, we formulate cyclic proof systems for reverse Hoare logic, observing that the soundness conditions required to make them work correctly are natural analogues of their equivalents in standard Hoare logic.  Lastly, our essentially uniform translation of standard Hoare logic and reverse Hoare logic proofs into cyclic proofs is quite pleasant and may provide, we hope, some minor technical interest.

Potential directions for future work might include, for example, the implementation of our cyclic proof systems within a suitable platform such as the Cyclist theorem prover~\cite{Brotherston-Gorogiannis-Petersen:12}, or the extension of existing systems of incorrectness logic with our cyclic proof principles~\cite{OHearn:19,Le:OOPSLA:2022}.


%

\end{document}